\documentclass[11pt]{article}

\usepackage[margin=1.15in]{geometry}
\usepackage{amsmath,amssymb,amsthm,mathtools}
\usepackage{microtype}
\usepackage{algorithm}
\usepackage[noend]{algpseudocode}
\usepackage[numbers,sort&compress]{natbib}
\usepackage[hidelinks]{hyperref}
\usepackage[nameinlink,capitalise]{cleveref}

\hypersetup{
  pdftitle={Exact Hill Shares Are Simultaneous Guarantees},
  pdfauthor={Bo Li, Hanyu Li, Chenghua Liu},
  pdfsubject={Fair allocation of indivisible bads},
  pdfkeywords={Hill's share, indivisible bads, maximin share, polynomial-time algorithm}
}

\allowdisplaybreaks

\theoremstyle{plain}
\newtheorem{theorem}{Theorem}[section]
\newtheorem{lemma}[theorem]{Lemma}
\newtheorem{proposition}[theorem]{Proposition}
\newtheorem{corollary}[theorem]{Corollary}

\theoremstyle{definition}

\makeatletter
\providecommand{\theHALG@line}{}
\renewcommand{\theHALG@line}{\thealgorithm.\arabic{ALG@line}}
\makeatother

\theoremstyle{remark}

\newcommand{\MMS}{\operatorname{MMS}}

\newcommand{\bucket}{\operatorname{bucket}}

\title{\textbf{Exact Hill's Shares Are Simultaneous Guarantees}}
\author{Bo Li\textsuperscript{1}\thanks{\href{mailto:comp-bo.li@polyu.edu.hk}{\texttt{comp-bo.li@polyu.edu.hk}}}\quad
Hanyu Li\textsuperscript{2}\thanks{\href{mailto:lhydave@pku.edu.cn}{\texttt{lhydave@pku.edu.cn}}}\quad
Chenghua Liu\textsuperscript{3}\thanks{\href{mailto:liuch.russell@gmail.com}{\texttt{liuch.russell@gmail.com}}}\\[0.5em]
\textsuperscript{1}Department of Computing, The Hong Kong Polytechnic University\\[-0.1em]
\textsuperscript{2}CFCS, School of Computer Science, Peking University\\[-0.1em]
\textsuperscript{3}Institute of Software, Chinese Academy of Sciences}
\date{}

\begin{document}

\maketitle

\begin{abstract}
Fair division of indivisible bads seeks allocations that guarantee every agent a bundle whose cost is no larger than a meaningful fairness benchmark. The canonical minimax share has widely been used; unfortunately, it is not a simultaneous guarantee. Hill (\emph{Ann. Probab.}, 1987) initiated a complementary approach in which the share depends only on the number of agents and the
largest possible single-item value. Li et al. (\emph{ACM Trans. Econ. Comput.}, 2024) gave the exact Hill formula and proved that its {\em monotone closure} is a simultaneous guarantee. The closure treats the largest-item cost only as an upper bound and can be strictly larger than the share conditioned on the \emph{actual} largest item.  They asked whether this smaller exact share is itself simultaneously guaranteed for three or more agents. We resolve this open question affirmatively.  For any number of agents and arbitrary heterogeneous largest-item costs, there is one allocation that satisfies every agent's exact Hill's share. We further provide a polynomial-time algorithm computes such an allocation. Our algorithm combines an ordered moving knife with a tail-domination invariant, and a one-sided trimmed subset-sum routine for the two-agent endpoint without computing an exact minimax partition.
\end{abstract}

\section{Introduction}

The task of fair division is to allocate objects among agents with equal
rights but different preferences \citep{Steinhaus1949}.  We consider the
standard additive model for indivisible bads
\citep{AzizRaucheckerSchryenWalsh2017,LiMoulinSunZhou2024}.  Let
\(N\) be the set of \(n\) agents, labeled \(1,\ldots,n\), and \(M\) be
the set of bads.  Agent \(i\) has a nonnegative additive disutility \(v_i\).
We scale each agent's disutility so that receiving all the bads would cost 1, i.e., \(v_i(M)=1\).
An allocation \((A_1,\ldots,A_n)\) partitions \(M\), agent \(i\) receives bundle $A_i$ and incurs
cost \(v_i(A_i)\).  Following the guaranteed-share viewpoint, we seek a
personal cost bound, depending only on \(v_i\) and \(n\), that one allocation
can meet simultaneously for all agents.

If the bads were divisible, each agent could be charged at most her
proportional share \(1/n\).  Indivisibility makes this impossible when a
single chore carries a large fraction of an agent's total disutility.  The
minimax-share criterion for bads captures this issue
\citep{Budish2011,AzizRaucheckerSchryenWalsh2017}.  Agent \(i\) imagines
partitioning the bads into \(n\) bundles and then receiving the most costly
bundle; she chooses a partition minimizing that cost.
Formally, let \(\Pi_n(M)\) be the set of all partitions of \(M\) into \(n\) bundles.
Then this share is defined as
\[
  \MMS_n(v_i)=
  \min_{(P_1,\ldots,P_n)\in \Pi_n(E)}\max_{1\le r\le n}v_i(P_r).
\]
The inner maximum selects the most costly bundle, and the outer
minimum selects the partition that makes this cost as small as possible.
This divide-and-choose share depends on the full valuation, is NP-hard to
compute exactly, and cannot always be guaranteed simultaneously for
heterogeneous agents
\citep{Budish2011,GareyJohnson1979,AzizRaucheckerSchryenWalsh2017,
FeigeSapirTauber2021}.

An earlier line of work asks for simpler guarantees that depend only on the
total value and the value of the largest item.  Let \(\alpha\) be the
largest value that a single item can have.  Hill's
probability-measure theorem gives every agent a share \(V_n(\alpha)\) when
every atom has value at most \(\alpha\) \citep{Hill1987}, where $n$ is the number of agents.
Demko and Hill formulated the corresponding indivisible-object allocation
problem \citep{DemkoHill1988}.  For goods,
\citet{MarkakisPsomas2011} allowed agents to have different largest-item
values and gave a polynomial-time allocation algorithm.  Gourv\`es, Monnot,
and Tlilane then showed that Hill's monotone guarantee is not tight and
studied the finer model in which \(\alpha\) is the \emph{actual} largest-item
value \citep{GourvesMonnotTlilane2015}.

This exact-maximum model was brought to indivisible bads in
\citep{LiMoulinSunZhou2024}.  Fix a possible cost \(\alpha\in(0,1]\) for the
most costly single bad, and let
\(\mathcal V(\alpha)\) contain all normalized additive disutilities (i.e., $v_i(M)=1$ and $v_i(S)=\sum_{e\in S}v_i(e)$ for all $S \subseteq M$), over
arbitrary finite item sets, whose largest item has cost exactly \(\alpha\), i.e., $\alpha=\max_{e\in M}v(e)$.
The exact Hill's share asks for the largest minimax share that can arise in
this class;  
We denote this value by
\[
  \Delta_n^+(\alpha)
  =\sup_{v\in\mathcal V(\alpha)}\MMS_n(v).
\]
A closed formula of $\Delta_n^+(\alpha)$ is given in \citep{LiMoulinSunZhou2024}.  For a given instance,
let \(\alpha_i=\max_{e\in M}v_i(e)\) be agent \(i\)'s largest single-bad
cost.  Because the exact share is not monotone in this value, they also
considered its {\em monotone closure}
\[
  V_n(\alpha)=\sup_{\beta\le\alpha}\Delta_n^+(\beta).
\]
Here \(V_n(\alpha)\) covers every possible actual largest-item cost at most
\(\alpha\).  They proved that every instance admits an allocation such that for every $i=1,\ldots,n$, $v_i(A_i)\le V_n(\alpha_i)$. 
That is, the monotone closure of Hill's share is a simultaneous guarantee.
Since \(V_n(\alpha_i)\ge\Delta_n^+(\alpha_i)\), this result does not guarantee
the exact shares.  
They left open whether the exact Hill's share (i.e., the  stronger inequalities with
\(\Delta_n^+(\alpha_i)\)) can hold simultaneously for three or more agents.

\medskip

The current paper resolves this open problem affirmatively.

\begin{theorem}
\label{thm:main}
Every normalized nonnegative additive instance with \(n\ge2\) admits an allocation \((A_1,\ldots,A_n)\) of \(M\) satisfying
\[
  v_i(A_i)\le \Delta_n^+(\alpha_i)
  \qquad\text{for every }i=1,\ldots,n.
\]
Such an allocation can be found in $O(m^3)$ time. 
\end{theorem}

The earlier proof for the monotone closure \(V_n\) follows a recursive
template \citep{LiMoulinSunZhou2024}.  Its goal in each step is to finish the
allocation for one agent without destroying the guarantees promised to
everyone else.  An ordered moving knife grows a bundle along the bads, assigns
a safe bundle to one agent, and then removes both that bundle and that agent.
The same step is repeated on the smaller instance.  Monotonicity allows the
proof to compare each remaining agent's target before and after the removal,
which is what makes the recursion valid.

The same argument does not directly give Theorem~\ref{thm:main}.  For the
exact share \(\Delta_n^+\), removing a bundle changes the number of agents,
the remaining total cost, and the largest remaining bad, and the new target
can move in either direction.  Our first technical contribution is a tail-domination
lemma that replaces the missing monotonicity.  It proves that the peeled
bundle removes enough of the costly end of the instance to compensate for
having one fewer agent.  The moving-knife recursion therefore remains valid,
and one induction proves the exact simultaneous guarantee.

The second contribution  is the polynomial-time implementation.  For the exact share,
the known two-agent argument obtains existence from an exact minimax
partition \citep{LiMoulinSunZhou2024}.  Computing such a partition can encode
the NP-hard \textsc{Partition} problem \citep{GareyJohnson1979}.  We show that
the exact Hill formula leaves enough room to use a sufficiently balanced
split instead.  A trimmed subset-sum routine finds such a partition in polynomial
time, completing the constructive proof.

This paper is organized as follows.  Section~\ref{sec:related} surveys
related work.  
Section~\ref{sec:framework} develops the ordered recursive
framework.  Section~\ref{sec:tail} proves the tail-domination inequality,
and Section~\ref{sec:two-agent-base} gives the polynomial two-agent base.
Section~\ref{sec:algorithm} combines these parts into the allocation algorithm
and its proof.  The final section concludes.

\section{Related work}
\label{sec:related}

Several familiar fairness notions are tests on a completed allocation rather
than worst-case shares.  Proportionality up to one object (Prop1) and
proportionality up to any object (PropX) allow one object to be removed or
added \citep{AzizMoulinSandomirskiy2020,Moulin2019}; envy-freeness up to one
item (EF1) and envy-freeness up to any item (EFX) similarly relax envy-freeness
\citep{LiptonMarkakisMosselSaberi2004,Budish2011,CaragiannisEtAl2019}.  These
conditions are easy to verify after an allocation is produced, but they do
not specify the kind of numerical guarantee that an agent can compute from
her own disutility before the allocation.  The truncated proportional share
is another efficiently verifiable target, but for goods it can be stronger
than the maximin share and therefore cannot always be guaranteed
\citep{BabaioffEzraFeige2022}.

Results for goods also do not transfer to bads by simply changing signs: even
the one-object relaxations above behave differently in the two settings.  We
refer to \citet{Moulin2019} and \citet{AmanatidisEtAl2023} for broader surveys
of fair division with indivisible objects.  For maximin shares in particular,
exact simultaneous feasibility may fail for chores
\citep{AzizRaucheckerSchryenWalsh2017,FeigeSapirTauber2021}.

The closest work derived the exact Hill's share for bads, guaranteed its
monotone closure, and left the simultaneous exact-share question open for
\(n\ge3\) \citep{LiMoulinSunZhou2024}.  Our theorem resolves this question.
Technically, we combine the standard ordered-instance reduction
\citep{BouveretLemaitre2016,HuangLu2021} with the new tail-domination argument
and the polynomial two-agent construction described in the introduction.

\section{The recursive allocation framework}
\label{sec:framework}

In this section, we introduce the recursive allocation framework of our algorithm, and show how our algorithm differs from that in \citep{LiMoulinSunZhou2024}.


\subsection{Ordered instances reduction}
\label{sec:ordered}

The ordered-instance reduction is standard in maximin-share allocation
\citep{BouveretLemaitre2016,HuangLu2021}; it is also the first step of the
algorithm in \citep{LiMoulinSunZhou2024}.  We state only the form
used here.

For each agent \(i\), sort her item costs from largest to smallest:
\[
  w_i(1)\ge w_i(2)\ge\cdots\ge w_i(m).
\]
Rank \(r\) now means agent \(i\)'s \(r\)-th most costly item.  These ranks
need not refer to the same original item for different agents; they form an
abstract instance in which all agents nevertheless share the order
\(1,\ldots,m\).  Agent \(i\)'s cost for rank \(r\) is \(w_i(r)\).

\begin{lemma}[Ordered-instance reduction~\citep{BouveretLemaitre2016,HuangLu2021}]
\label{lem:lifting}
Every allocation of the ordered rank items can be used to produce an allocation of
the original items without increasing any agent's cost.
\end{lemma}

The standard reduction processes ranks from last to first and assigns a
cheapest remaining original item at each step.  With sorted item lists and
moving pointers, it takes \(O(nm)\) time after sorting.

\subsection{The recursive prefix step}
\label{sec:recursive-step}

The moving-knife recursion follows the prior work
\citep{LiMoulinSunZhou2024}, but uses a different threshold.  Suppose an agent
sees a remaining item set of total cost \(T>0\), with largest single-item cost
\(p\), and \(q\ge3\) agents remain.  We use the exact Hill's share for this
actual residual state, written on the original scale as
\[
  B_q(T,p)=T\Delta_q^+(p/T).
  \tag{1}\label{eq:scaled-hill}
\]
For an all-zero remaining set, let \(B_q(0,0)=0\).

Now consider a common ordered suffix beginning at rank \(\ell\).  For each
active agent \(i\), let \(T_i\) be her total cost for this suffix, let
\(p_i=w_i(\ell)\) be her largest remaining item cost, and let
\(B_i=B_q(T_i,p_i)\) be her current limit.  Starting at rank \(\ell\), grow a
prefix until it first becomes too costly for that agent.  Formally, let
\[
  r_i=\min\left\{r\ge\ell:
    \sum_{t=\ell}^{r}w_i(t)>B_i
  \right\}.
\]
Thus the prefix ending immediately before \(r_i\) costs at most \(B_i\) to
agent \(i\), while including rank \(r_i\) would exceed her limit.

Choose an agent \(a\) whose crossing rank \(r_a\) is largest, give her ranks
\(\ell,\ldots,r_a-1\), and recurse on the suffix beginning at \(r_a\) with
one fewer agent.  Agent \(a\) is safe by the definition of \(r_a\).  For
every remaining agent \(j\), maximality of \(r_a\) gives \(r_j\le r_a\), so
\[
  \sum_{t=\ell}^{r_a}w_j(t)>B_j.
\]
The removed prefix consists of ranks \(\ell,\ldots,r_a-1\); rank \(r_a\) is
the largest item of the residual suffix.  Tail domination, proved next,
says that this crossing forces agent \(j\)'s new \((q-1)\)-agent
Hill bound to be at most the old limit \(B_j\).  This is exactly the fact
needed to repeat the same step recursively.

In the prior algorithm, the same geometric step is run with the larger
monotone share \(V_q\).  Its analysis proves that, after deleting the chosen
prefix and renormalizing the suffix, every survivor's new \(V_{q-1}\) limit
fits within her old \(V_q\) limit
\citep{LiMoulinSunZhou2024}.  That argument does not
establish the corresponding comparison for
\(\Delta_q^+\), which can move in either direction when the residual
largest-item ratio changes.  Lemmas~\ref{lem:ordinary-tail} and
\ref{lem:exceptional-tail} are our replacement for precisely this step.
Degenerate all-zero suffixes and suffixes with only one positive-cost item
are handled directly in the full algorithm.

\section{Tail domination}
\label{sec:tail}

The prefix step reduces the recursion to one comparison: after a strict
crossing, the exact Hill bound of the residual tail with one fewer agent must
not exceed the old bound.  
We prove this comparison in this section.  

\subsection{Exact-formula tools}
\label{sec:formulas}

For three or more agents, we use the following exact formula.

\begin{proposition}
\label{prop:qge3formula}
Let \(T>0\), \(0<p\le T\), and
\[
  k=\left\lfloor\frac{T-p}{qp}\right\rfloor.
\]
Then
\[
  B_q(T,p)=
  \max\left\{
    (k+1)p,\
    \frac{k+2}{q(k+1)}(T-p)
  \right\}.
  \tag{2}
\]
\end{proposition}

The two terms in (2) capture the two ways a residual instance can be hard:
several large items may already force a high cost, or the remaining mass may
be too large to spread evenly.  Tail domination must control both terms
after a prefix is removed.

The integer \(k\) depends on the ratio \(p/T\), and this ratio changes after
every removal.  Our tail proof packages the known formula into an upper
envelope indexed by an integer \(h\), avoiding the need to locate every new
ratio in its exact formula range.  Put
\[
  g_{q,h}(x)=
  \max\left\{
    (h+1)x,\
    \frac{h+2}{q(h+1)}(1-x)
  \right\}.
\]
For
\[
  k=\left\lfloor\frac{1-x}{qx}\right\rfloor,
\]
applying Proposition~\ref{prop:qge3formula} with \(T=1\) and \(p=x\)
and using \eqref{eq:scaled-hill} gives
\[
  \Delta_q^+(x)=B_q(1,x)=g_{q,k}(x).
\]
We next show that every member of the envelope family dominates this value:
\[
  \Delta_q^+(x)\le g_{q,h}(x)\qquad(h\ge0).
\]
Indeed, the two branches of \(g_{q,h}\) meet at
\[
  x=\frac{h+2}{q(h+1)^2+h+2},
\]
and adjacent envelopes meet at \(x=1/(q(h+1)+1)\).  If \(h<k\), the
decreasing branch for \(h\) dominates both branches for \(k\); if \(h>k\),
the increasing branch for \(h\) does so.  Endpoint values agree.

After scaling, define
\[
  G_{q,h}(T,p)=
  \max\left\{
    (h+1)p,\
    \frac{h+2}{q(h+1)}(T-p)
  \right\}.
  \tag{3}
\]
Then
\[
  B_q(T,p)\le G_{q,h}(T,p)
  \qquad\text{for every }h\ge0.
  \tag{4}
\]
Every member of this family has the same two branches as the exact formula
and dominates the exact share.  We may therefore choose \(h\) from the old
state: one choice controls a small new largest item, and another controls a
large one.  This freedom is the key algebraic device in the tail argument.

For \(q=2\), the exact formula has several additional ranges.  We need these
ranges for the transition from three agents to two and later to quantify the
slack in the terminal two-agent routine.  Write
\(d(\alpha)=\Delta_2^+(\alpha)\).

\begin{proposition}
\label{prop:q2formula}
For \(0<\alpha\le1\),
\[
d(\alpha)=
\begin{cases}
\max\{\alpha,1-\alpha\},
  & 1/3\le\alpha\le1,\\
2\alpha,
  & 2/7\le\alpha\le1/3,\\
(2+3\alpha)/5,
  & 7/27\le\alpha\le2/7,\\
3(1-\alpha)/4,
  & 1/5\le\alpha\le7/27,\\
\max\{(k+1)\alpha,(k+2)(1-\alpha)/(2(k+1))\},
  & 0<\alpha\le1/5.
\end{cases}
\tag{5}
\]
In the last line, \(k=\lfloor(1/\alpha-1)/2\rfloor\ge2\).
The formulas agree at every shared endpoint.
\end{proposition}

\subsection{Four or more agents}

We first prove the comparison required by the recursive prefix step when at
least four agents remain.  Normalize
one remaining agent's current total cost to one.  Her old bound is
\(d=\Delta_n^+(\alpha)\).  A prefix of cost \(C\) is removed, the next item
has cost \(p\), and the residual total is \(T=1-C\).  The crossing condition
is \(C+p>d\).  Tail domination says that the new \((n-1)\)-agent bound
\(T\Delta_{n-1}^+(p/T)\) is at most \(d\).

This comparison does not follow from monotonicity because the exact share is
nonmonotone in its largest-item parameter.  We prove it directly for four or
more agents and then handle the exceptional transition from three agents to
two.

\begin{lemma}[Tail domination for \(n\ge4\)]
\label{lem:ordinary-tail}
Let \(n\ge4\), and consider a nonincreasing normalized sequence
\[
  a_1=\alpha\ge a_2\ge\cdots\ge a_m\ge0.
\]
Put \(d=\Delta_n^+(\alpha)\).  Suppose a prefix of \(s\) items has cost
\(C\), the next item is \(p=a_{s+1}\), and
\[
  C+p>d.
  \tag{6}
\]
If \(T=1-C>0\), then
\[
  T\Delta_{n-1}^+(p/T)\le d.
\]
\end{lemma}

\begin{proof}[Proof sketch]
The exact formula has two sources of difficulty: several large items may
force a high load, or too much residual mass may remain to be spread among
the agents.  The two branches of \(G_{n-1,h}\) keep these obstructions
separate.

Let \(k\) be the index of the old state.  A strict crossing cannot occur
among only \(k+1\) items, because each costs at most \(\alpha\) and the old
bound satisfies \(d\ge(k+1)\alpha\).  Hence the removed prefix contains the
largest item and at least \(k\) further items, each no smaller than the next
item \(p\).  This converts the order of the items into a lower bound on the
mass already removed.

We then choose the envelope according to the size of \(p\).  If
\(p\le d/(k+2)\), take \(h=k+1\): the large-item branch is immediately at
most \(d\), while the crossing inequality leaves sufficiently little total
mass for the other branch.  If \(p>d/(k+2)\), take \(h=k\): the large-item
branch is controlled by \(p\le\alpha\), and the \(k\) removed items of size
at least \(p\) force enough mass to have disappeared to control the other
branch.  Thus a larger new maximum is paid for by a smaller residual total.
The exact-formula calculation only verifies that the constants in these two
charges meet at the threshold \(p=d/(k+2)\).  The details appear in
Appendix~\ref{app:ordinary-tail-proof}.
\end{proof}

\subsection{The transition from three agents to two}

The transition from three agents to two uses the exceptional part of the
two-agent formula and requires separate treatment.

\begin{lemma}[Tail domination for \(n=3\)]
\label{lem:exceptional-tail}
The conclusion of Lemma~\ref{lem:ordinary-tail} also holds for \(n=3\).
\end{lemma}

\begin{proof}[Proof sketch]
The transition to two agents has the same structure as the preceding lemma;
the only new issue is that the exact two-agent formula has several extra
ranges.  For \(k\ge2\), all of those ranges still lie below the same
two-branch envelope \(G_{2,h}\) for every \(h\ge2\).  Once this domination
is established, the identical switch at \(p=d/(k+2)\) applies: a small
\(p\) directly controls the large-item branch, whereas a large \(p\)
certifies that the removed prefix carried away enough mass.

The indices \(k=0\) and \(k=1\) correspond to instances with a genuinely
large maximum item.  There the tail can be handled directly by putting its
largest item alone when necessary and otherwise balancing greedily between
two bundles.  Thus the exceptional pieces of the two-agent formula do not
create a new recursive phenomenon; they are used only to certify the common
envelope at the terminal transition.  The detailed verification appears in
Appendix~\ref{app:exceptional-tail-proof}.
\end{proof}

\begin{corollary}
\label{cor:scaled-tail}
For an unnormalized ordered suffix with total \(T_0>0\), maximum \(p_0\),
and current bound \(B_q(T_0,p_0)\), remove a prefix of cost \(C\).  If the
removed prefix plus the next item exceeds \(B_q(T_0,p_0)\), then the
\((q-1)\)-agent Hill bound of the residual tail is at most the old bound.
This holds for every \(q\ge3\), using
Lemma~\ref{lem:ordinary-tail} when \(q\ge4\) and
Lemma~\ref{lem:exceptional-tail} when \(q=3\).
\end{corollary}

\section{Polynomial two-agent terminal routine}
\label{sec:two-agent-base}

The prior work already proves the exact two-agent guarantee by
divide-and-choose \citep{LiMoulinSunZhou2024}.  The divider computes
an exact minimax partition, so both bundles meet her exact Hill's share; the
other agent chooses her cheaper bundle, whose cost is at most half of her
residual total.  This establishes existence, but computing the exact
partition can encode \textsc{Partition} \citep{GareyJohnson1979}.  Our
contribution in this section is that we replace the exact minimax computation by a polynomial-time trimmed
subset-sum routine.

\subsection{The feasible interval}

Suppose the divider assigns nonnegative costs
\(c_1,\ldots,c_s\) to the residual items.  Let \(T>0\) be their total cost, \(p\) be the largest item cost, and $D=B_2(T,p)$
be her exact Hill's share.  A subset of cost \(x\) defines two bundles of
costs \(x\) and \(T-x\).  Both costs are at most \(D\) exactly when
\[
  x\in[L,U]=[T-D,D].
\]
Rather than computing the exact minimax partition, we only need to hit the
interval \([L,U]\).  The exact two-agent formula shows that this interval has
enough width to permit polynomial trimming.

To quantify the available slack, normalize by \(T\), put
\(\alpha=p/T\), and write
\[
  w(\alpha)=2d(\alpha)-1.
\]
This is the width of the normalized feasible interval.  
Although the exact formula is
known, the following width bound and the trimmed construction are new.

\begin{lemma}
\label{lem:width}
If \(0<\alpha\le1/3\), then
\[
  w(\alpha)\ge\frac37\alpha.
\]
\end{lemma}

\begin{proof}
In the range \(0<\alpha\le1/5\), let \(k\ge2\) be defined as in
Proposition~\ref{prop:q2formula}.  Dividing \(2d(\alpha)-1\) by \(\alpha\)
gives
\[
  \frac{w(\alpha)}{\alpha}
  =
  \max\left\{
    2(k+1)-\frac1\alpha,\
    \frac{1/\alpha-k-2}{k+1}
  \right\}.
\]
The first expression increases with \(\alpha\), the second decreases, and
their common value at the intersection is
\[
  \frac{k}{k+2}\ge\frac12.
\]

Across the other three ranges below \(1/3\), the corresponding ratios are
\[
  \frac{1/\alpha-3}{2},\qquad
  \frac{6-1/\alpha}{5},\qquad
  4-\frac1\alpha.
\]
Their minima on their respective intervals are
\(3/7,3/7,1/2\).  The value \(3/7\) is attained at
\(\alpha=7/27\).
\end{proof}

Returning to the divider's original scale, let \(W=U-L=2D-T\).  When
\(p/T\le1/3\), scaling Lemma~\ref{lem:width} and using \(p\ge T/s\) yields
\[
  W\ge\frac37p\ge\frac{3T}{7s}.
  \tag{7}
\]
The feasible interval is therefore an inverse-polynomial fraction of the
total cost.  This width is the source of the polynomial state bound.

If \(p/T\ge1/3\), put one maximum item in the first bundle and all remaining
items in the second.  Formula (5) gives
\[
  D=\max\{p,T-p\},
\]
so both bundles are feasible.

\subsection{One-sided trimmed subset sum}

It remains to treat \(p/T<1/3\).  Define
\[
  \eta=\frac{W}{2s},
  \tag{8}
\]
and place every exact sum \(z\) in bucket
\[
  \bucket(z)=\left\lfloor\frac{z}{\eta}\right\rfloor.
\]
An ordinary subset-sum dynamic program may contain exponentially many exact
sums.  We keep only one exact representative per bucket.  A feasible sum
below \(T/2\) must be approximated upward, whereas one above \(T/2\) must be
approximated downward.  We therefore run the following dynamic program
twice, retaining respectively the largest and the smallest representative
of every bucket.

\begin{algorithm}[H]
\caption{One-sided trimmed subset sum}\label{alg:trim}
\begin{algorithmic}[1]
\Procedure{Trim}{$\mathit{direction}$}
  \State $\mathcal T\gets\{0\mapsto0\}$
  \For{each item cost $c$}
    \State $\mathcal C\gets
      \{z,z+c:z\in\operatorname{image}(\mathcal T)\}$
    \State $\mathcal T'\gets\varnothing$
    \For{each bucket $b$ occupied by $\mathcal C$}
      \If{$\mathit{direction}=\mathrm{MAX}$}
        \State $z_b\gets\max\{z\in\mathcal C:\bucket(z)=b\}$
      \Else
        \State $z_b\gets\min\{z\in\mathcal C:\bucket(z)=b\}$
      \EndIf
      \State Store $z_b$ in $\mathcal T'[b]$, together with its predecessor
        and include/exclude bit
    \EndFor
    \State $\mathcal T\gets\mathcal T'$
  \EndFor
  \State \Return $\mathcal T$
\EndProcedure
\end{algorithmic}
\end{algorithm}

Scan both final tables.  If a stored exact sum belongs to \([L,U]\), follow
its predecessor pointers to reconstruct the subset.
The next lemma formalizes the two one-sided errors while ensuring that every
retained value remains an actual subset sum.

\begin{lemma}
\label{lem:trim-invariants}
After processing \(t\) items:
\begin{enumerate}
  \item for every exactly reachable sum \(x\), the MAX table contains an
  actual subset sum \(y\) with
  \[
    x\le y\le x+t\eta;
  \]
  \item for every exactly reachable sum \(x\), the MIN table contains an
  actual subset sum \(z\) with
  \[
    x-t\eta\le z\le x.
  \]
\end{enumerate}
\end{lemma}

\begin{proof}
The claim is immediate for \(t=0\).  For the MAX induction step, represent
the predecessor of \(x\), add the current item if \(x\) uses it, and consider
the resulting candidate.  Keeping the maximum candidate in its bucket can
only move the value upward, by less than one bucket width \(\eta\).  This
adds at most \(\eta\) to the accumulated error.  The MIN proof is symmetric.
Every retained number remains an exact subset sum.
\end{proof}

\begin{lemma}
\label{lem:trim-finds}
At least one final table contains a stored exact sum in \([L,U]\).
\end{lemma}

\begin{proof}
For the given valuation,
\[
  \MMS_2(c)\le\Delta_2^+(p/T)\,T=D.
\]
Thus some true partition has maximum load at most \(D\), and one of its
subset sums \(x\) belongs to \([L,U]\).  This fact follows directly from the
definition of the worst-case Hill's share; it does not use the simultaneous
allocation theorem.

If \(x\le T/2\), the MAX invariant at \(t=s\) gives
\[
  L\le x\le y
  \le x+\frac{W}{2}
  \le\frac{T}{2}+\frac{W}{2}
  =U.
\]
If \(x\ge T/2\), the MIN invariant gives
\[
  L=\frac{T}{2}-\frac{W}{2}
  \le x-\frac{W}{2}
  \le z\le x\le U.
\]
This includes the case in which the only feasible witness lies exactly at an
endpoint.
\end{proof}

\begin{proposition}
\label{prop:two-bin-time}
The two-bin construction uses \(O(s^3)\) exact rational operations,
\(O(s^2)\) rolling working states, and \(O(s^3)\) predecessor words.
\end{proposition}

\begin{proof}
By (7) and (8),
\[
  \frac{T}{\eta}+1
  =
  \frac{2sT}{W}+1
  \le
  \frac{14}{3}s^2+1.
\]
Hence each layer has \(O(s^2)\) buckets and each table has \(s\) layers.
The buckets can be stored in a deterministic array of length
\(\lfloor T/\eta\rfloor+1\), or in a deterministic balanced map.  The state
bound is independent of a common denominator or the numerical magnitude of
the input.
\end{proof}

\subsection{Divide-and-choose}

We can now state the two-agent allocation formally.

\begin{algorithm}[H]
\caption{Polynomial divide-and-choose}\label{alg:two-agent}
\begin{algorithmic}[1]
\Require Two agents and a common ordered item set
\Ensure An allocation of the ordered item set to the two agents
\State Choose either agent as the divider
\State Run Algorithm~\ref{alg:trim} on the divider's costs in both directions
\State Scan the two final tables for a sum $z\in[L,U]$
\State Reconstruct the corresponding subset $P$ and let $Q$ be its complement
\State Let the other agent choose her cheaper bundle from $P$ and $Q$
\State Give the divider the remaining bundle
\end{algorithmic}
\end{algorithm}

The divider is safe because both bundles meet her bound.  The chooser pays
at most half of her total cost for the item set, and
\[
  \Delta_2^+(\alpha)\ge\frac12
\]
for every \(\alpha\).  Therefore Algorithm~\ref{alg:two-agent} will
serve as the polynomial base of the recursive allocation defined in
Section~\ref{sec:algorithm}.

\section{Proof of the main theorem}
\label{sec:algorithm}

We now combine the preceding parts in one constructive induction.  We work
on the ordered instance, repeatedly assign the safe prefix described in
Section~\ref{sec:recursive-step}, and use tail domination to protect every
remaining agent.  The polynomial two-agent routine is the base case.  The
same induction proves existence, correctness, and termination of the
algorithm.

\begin{proof}[Proof of Theorem~\ref{thm:main}]
Sort every agent's costs into the ordered instance and precompute suffix sums
\[
  S_i(r)=\sum_{t=r}^{m}w_i(t).
\]
The recursive call \(\textsc{Allocate}(A,\ell)\) receives an active agent set
\(A\), \(q=|A|\), and the common rank suffix
\(\{\ell,\ldots,m\}\).

The invariant gives the meaning of a recursive call.  Every active agent
\(i\) must ultimately receive cost at most the exact \(q\)-agent Hill bound
of her current suffix:
\[
  B_q(S_i(\ell),w_i(\ell)).
\]
The following algorithm implements the prefix step while preserving this
invariant.\par\smallskip

\begin{algorithm}[H]
\caption{Recursive allocation on the ordered instance}
\label{alg:ordered-allocation}
\begin{algorithmic}[1]
\Procedure{Allocate}{$A,\ell$}
  \If{$\ell>m$}
    \State \Return
  \ElsIf{$\exists i\in A$ such that $S_i(\ell)=0$}
    \State Assign the entire suffix to such an $i$
    \State \Return
  \ElsIf{$\exists i\in A$ such that $w_i(\ell)=S_i(\ell)$}
    \State Assign the entire suffix to such an $i$
    \State \Return
  \ElsIf{$|A|=2$}
    \State Apply Algorithm~\ref{alg:two-agent} to
      $A$ and $\{\ell,\ldots,m\}$
    \State \Return
  \EndIf
  \State $q\gets|A|$
  \For{each $i\in A$}
    \State $T_i\gets S_i(\ell)$ and $p_i\gets w_i(\ell)$
    \State $k_i\gets\lfloor(T_i-p_i)/(qp_i)\rfloor$
    \State $B_i\gets
      \max\{(k_i+1)p_i,(k_i+2)(T_i-p_i)/(q(k_i+1))\}$
    \State $r_i\gets
      \min\{r\ge\ell:\sum_{t=\ell}^{r}w_i(t)>B_i\}$
  \EndFor
  \State Choose $a\in\arg\max_{i\in A}r_i$
  \State Assign ranks $\ell,\ldots,r_a-1$ to $a$
  \State \Call{Allocate}{$A\setminus\{a\},r_a$}
\EndProcedure
\State Run \Call{Allocate}{$\{1,\ldots,n\},1$}
\State Lift ranks $m,m-1,\ldots,1$ back to the original items
\end{algorithmic}
\end{algorithm}

The crossing comparison is strict: the allocated prefix stops immediately
before the first item that would violate the current bound.  Since ordered prefix costs are
nondecreasing, each \(r_i\) is found by binary search using suffix sums.
The two early rules are valid because \(B_q(0,0)=0\) and
\(B_q(T,T)=T\).  Otherwise \(0<p_i<T_i\), and for \(q\ge3\) both branches
of the exact formula (2) are strictly smaller than \(T_i\).  Hence the
whole suffix crosses \(B_i\), so every \(r_i\) exists.

We prove the invariant by induction on the number \(q\) of active agents.
With \(q=2\), either an
early rule is immediately safe or Algorithm~\ref{alg:two-agent} applies.
Suppose \(q\ge3\).
The selected agent \(a\) receives the prefix ending immediately before her
first strict crossing, so her cost is at most \(B_a\).  For each survivor
\(j\), maximality of \(r_a\) implies \(r_j\le r_a\); hence the prefix through
rank \(r_a\) exceeds \(B_j\).

The next recursive call first handles two degenerate suffixes.  If some
active agent has residual total zero, the first early rule gives her the
suffix at zero cost and gives everyone else nothing.  If some active agent
say \(j\) has positive residual total equal to its first item, the second
early rule gives her the suffix at cost
\[
  S_j(r_a)=w_j(r_a)\le w_j(\ell)\le B_j,
\]
and all other survivors receive nothing.

Otherwise every survivor has
\[
  0<w_j(r_a)<S_j(r_a).
\]
The strict crossing is exactly the hypothesis of
Corollary~\ref{cor:scaled-tail}, so the survivor's new
\((q-1)\)-agent bound is no larger than \(B_j\).  The induction hypothesis
therefore applies to the residual suffix and proves the invariant.

At the top call, \(S_i(1)=1\) and \(w_i(1)=\alpha_i\), so the invariant gives
the exact-maximum Hill bound for every agent.  Finally,
Lemma~\ref{lem:lifting} returns to the original items without increasing any
cost.  This proves the allocation claim.

It remains to bound the running time.  Sorting all \(n\) rows uses $O(nm\log m)$ comparisons, and suffix sums use \(O(nm)\) additions.  At a
recursive level with \(q\) active agents, thresholds and binary searches use
\(O(q\log m)\) exact operations.  The recursion is a single chain, with
one agent removed at every nonterminal level, so
\[
  \sum_{q=3}^{n}O(q\log m)
  =
  O(n^2\log m).
\]
The two-agent base uses \(O(m^3)\) operations by
Proposition~\ref{prop:two-bin-time}, and identical order reduction uses \(O(nm)\).
Hence the total running time is $O(m^3)$.
\end{proof}


\section{Conclusion}

The exact-maximum Hill's share is a simultaneous guarantee for every number of
agents and arbitrary heterogeneous exact maxima, resolving the open problem
in \citep{LiMoulinSunZhou2024}.  The same constructive induction computes the
allocation in polynomial time.  Tail domination makes the nonmonotone exact
share survive recursive prefix removal, and one-sided trimming prevents the
two-agent base from requiring an exact solution to \textsc{Partition}.  The
corresponding exact-maximum question for goods, as well as Pareto-efficient
and incentive-compatible variants for bads, remains open.

\section*{Acknowledgments}
OpenAI Codex assisted with proof development and verification.  The authors
independently checked all arguments and assume full responsibility for the
content of this manuscript.

\bibliographystyle{plainnat}
\bibliography{references}

@article{LiMoulinSunZhou2024,
  author  = {Bo Li and Herv{\'e} Moulin and Ankang Sun and Yu Zhou},
  title   = {On Hill's Worst-Case Guarantee for Indivisible Bads},
  journal = {ACM Transactions on Economics and Computation},
  volume  = {12},
  number  = {4},
  pages   = {14:1--14:27},
  year    = {2024},
  doi     = {10.1145/3703845}
}

@article{Hill1987,
  author  = {Theodore P. Hill},
  title   = {Partitioning General Probability Measures},
  journal = {The Annals of Probability},
  volume  = {15},
  number  = {2},
  pages   = {804--813},
  year    = {1987}
}

@article{DemkoHill1988,
  author  = {Stephen Demko and Theodore P. Hill},
  title   = {Equitable Distribution of Indivisible Objects},
  journal = {Mathematical Social Sciences},
  volume  = {16},
  number  = {2},
  pages   = {145--158},
  year    = {1988},
  doi     = {10.1016/0165-4896(88)90047-9}
}

@article{Steinhaus1949,
  author  = {Hugo Steinhaus},
  title   = {Sur la division pragmatique},
  journal = {Econometrica},
  volume  = {17},
  number  = {Supplement},
  pages   = {315--319},
  year    = {1949}
}

@article{Budish2011,
  author  = {Eric Budish},
  title   = {The Combinatorial Assignment Problem: Approximate Competitive
             Equilibrium from Equal Incomes},
  journal = {Journal of Political Economy},
  volume  = {119},
  number  = {6},
  pages   = {1061--1103},
  year    = {2011}
}

@article{AzizMoulinSandomirskiy2020,
  author  = {Haris Aziz and Herv{\'e} Moulin and Fedor Sandomirskiy},
  title   = {A Polynomial-Time Algorithm for Computing a Pareto Optimal and
             Almost Proportional Allocation},
  journal = {Operations Research Letters},
  volume  = {48},
  number  = {5},
  pages   = {573--578},
  year    = {2020},
  doi     = {10.1016/j.orl.2020.07.005}
}

@inproceedings{LiptonMarkakisMosselSaberi2004,
  author    = {Richard J. Lipton and Evangelos Markakis and Elchanan Mossel and
               Amin Saberi},
  title     = {On Approximately Fair Allocations of Indivisible Goods},
  booktitle = {Proceedings of the 5th ACM Conference on Electronic Commerce},
  pages     = {125--131},
  publisher = {ACM},
  year      = {2004}
}

@article{CaragiannisEtAl2019,
  author  = {Ioannis Caragiannis and David Kurokawa and Herv{\'e} Moulin and
             Ariel D. Procaccia and Nisarg Shah and Junxing Wang},
  title   = {The Unreasonable Fairness of Maximum Nash Welfare},
  journal = {ACM Transactions on Economics and Computation},
  volume  = {7},
  number  = {3},
  pages   = {12:1--12:32},
  year    = {2019}
}

@inproceedings{BabaioffEzraFeige2022,
  author    = {Moshe Babaioff and Tomer Ezra and Uriel Feige},
  title     = {On Best-of-Both-Worlds Fair-Share Allocations},
  booktitle = {Web and Internet Economics},
  series    = {Lecture Notes in Computer Science},
  volume    = {13778},
  pages     = {237--255},
  publisher = {Springer},
  year      = {2022}
}

@article{Moulin2019,
  author  = {Herv{\'e} Moulin},
  title   = {Fair Division in the Internet Age},
  journal = {Annual Review of Economics},
  volume  = {11},
  pages   = {407--441},
  year    = {2019},
  doi     = {10.1146/annurev-economics-080218-025559}
}

@article{AmanatidisEtAl2023,
  author  = {Georgios Amanatidis and Haris Aziz and Georgios Birmpas and
             Aris Filos-Ratsikas and Bo Li and Herv{\'e} Moulin and
             Alexandros A. Voudouris and Xiaowei Wu},
  title   = {Fair Division of Indivisible Goods: Recent Progress and Open
             Questions},
  journal = {Artificial Intelligence},
  volume  = {322},
  pages   = {103965},
  year    = {2023},
  doi     = {10.1016/j.artint.2023.103965}
}

@inproceedings{AzizRaucheckerSchryenWalsh2017,
  author    = {Haris Aziz and Gerhard Rauchecker and Guido Schryen and
               Toby Walsh},
  title     = {Algorithms for Max-Min Share Fair Allocation of Indivisible
               Chores},
  booktitle = {Proceedings of the Thirty-First AAAI Conference on Artificial
               Intelligence},
  pages     = {335--341},
  year      = {2017}
}

@inproceedings{FeigeSapirTauber2021,
  author    = {Uriel Feige and Ariel Sapir and Laliv Tauber},
  title     = {A Tight Negative Example for MMS Fair Allocations},
  booktitle = {Web and Internet Economics},
  series    = {Lecture Notes in Computer Science},
  volume    = {13112},
  pages     = {355--372},
  publisher = {Springer},
  year      = {2021}
}

@inproceedings{MarkakisPsomas2011,
  author    = {Evangelos Markakis and Christos-Alexandros Psomas},
  title     = {On Worst-Case Allocations in the Presence of Indivisible Goods},
  booktitle = {Web and Internet Economics},
  series    = {Lecture Notes in Computer Science},
  volume    = {7090},
  pages     = {278--289},
  publisher = {Springer},
  year      = {2011}
}

@article{GourvesMonnotTlilane2015,
  author  = {Laurent Gourv{\`e}s and J{\'e}r{\^o}me Monnot and Lydia Tlilane},
  title   = {Worst Case Compromises in Matroids with Applications to the
             Allocation of Indivisible Goods},
  journal = {Theoretical Computer Science},
  volume  = {589},
  pages   = {121--140},
  year    = {2015},
  doi     = {10.1016/j.tcs.2015.04.029}
}

@article{BouveretLemaitre2016,
  author  = {Sylvain Bouveret and Michel Lema{\^\i}tre},
  title   = {Characterizing Conflicts in Fair Division of Indivisible Goods
             Using a Scale of Criteria},
  journal = {Autonomous Agents and Multi-Agent Systems},
  volume  = {30},
  pages   = {259--290},
  year    = {2016}
}

@inproceedings{HuangLu2021,
  author    = {Xin Huang and Pinyan Lu},
  title     = {An Algorithmic Framework for Approximating Maximin Share
               Allocation of Chores},
  booktitle = {Proceedings of the 22nd ACM Conference on Economics and
               Computation},
  pages     = {630--631},
  publisher = {ACM},
  year      = {2021}
}

@book{GareyJohnson1979,
  author    = {Michael R. Garey and David S. Johnson},
  title     = {Computers and Intractability: A Guide to the Theory of
               NP-Completeness},
  publisher = {W. H. Freeman},
  year      = {1979}
}

\clearpage
\appendix

\section{Proof of Lemma~\ref{lem:ordinary-tail}}
\label{app:ordinary-tail-proof}

\begin{proof}
Set
\[
  k=\left\lfloor\frac{1-\alpha}{n\alpha}\right\rfloor.
\]
We bound the new quantity \(B_{n-1}(T,p)\) through (4), choosing \(h=k+1\)
when \(p\) is small and \(h=k\) when \(p\) is large.  First,
\[
  d=
  \max\left\{
    (k+1)\alpha,\
    \frac{k+2}{n(k+1)}(1-\alpha)
  \right\}.
  \tag{9}
\]
Since \(d\ge(k+1)\alpha\), condition (6) forces \(s\ge k+1\):
no \(k+1\) items can have total greater than \(d\).  The removed items are all
at least \(p\), so
\[
  C\ge\alpha+kp.
  \tag{10}
\]

The two branches in (9) meet at
\[
  \alpha_0=
  \frac{k+2}{n(k+1)^2+k+2}.
\]
For every \(\alpha\) that gives this value of \(k\), the maximum of the two
branches is at least their common value:
\[
  d_{\min}=
  \frac{(k+1)(k+2)}{n(k+1)^2+k+2}.
\]
Consequently,
\[
  d\ge\frac{k+3}{n(k+2)+1}.
  \tag{11}
\]
After clearing positive denominators, the required difference is
\((n-2)k+n-4\ge0\).

Set \(q=n-1\) in (3)--(4).  First suppose
\(p\le d/(k+2)\), and choose \(h=k+1\).  The increasing branch of
\(G_{q,k+1}(T,p)\) is at most \(d\).  Moreover, (6) gives
\[
  T-p=1-(C+p)<1-d.
\]
Inequality (11) is equivalent to
\[
  \frac{k+3}{(n-1)(k+2)}(1-d)\le d,
\]
so the decreasing branch is also at most \(d\).

Now suppose \(p>d/(k+2)\), and choose \(h=k\).  Since \(p\le\alpha\),
\[
  (k+1)p\le(k+1)\alpha\le d.
\]
Using (10),
\[
  C+p
  \ge\alpha+(k+1)p
  >
  \alpha+\frac{k+1}{k+2}d.
\]
The decreasing branch in (9) gives
\[
  1-\alpha\le\frac{n(k+1)}{k+2}d.
\]
Thus
\[
  T-p=1-(C+p)
  <
  1-\alpha-\frac{k+1}{k+2}d
  \le
  \frac{(n-1)(k+1)}{k+2}d.
\]
The decreasing branch of \(G_{n-1,k}(T,p)\) is at most \(d\).
Equation (4) completes both cases.
\end{proof}

\clearpage
\section{Proof of Lemma~\ref{lem:exceptional-tail}}
\label{app:exceptional-tail-proof}

\begin{proof}
Set
\[
  k=\left\lfloor\frac{1-\alpha}{3\alpha}\right\rfloor.
\]
Then
\[
  d=
  \max\left\{
    (k+1)\alpha,\
    \frac{k+2}{3(k+1)}(1-\alpha)
  \right\},
  \tag{12}
\]
and the crossing assumption implies \(s\ge k+1\) and
\(C\ge\alpha+kp\).

The two exceptional indices \(k=0\) and \(k=1\) admit direct two-bin
arguments.  The remaining indices follow the same envelope comparison as in
Lemma~\ref{lem:ordinary-tail}.

If \(k=0\), then
\[
  d=\max\{\alpha,2(1-\alpha)/3\}.
\]
Every possible tail has \(p\le d\) and
\(T\le1-\alpha\le3d/2\).  If \(p>d/2\), put the largest tail item alone;
the other bin has load \(T-p<d\).  If \(p\le d/2\), greedy balancing leaves
the two loads differing by at most \(p\), so the larger load is at most
\[
  \frac{T+p}{2}\le\frac{3d/2+d/2}{2}=d.
\]
Thus the tail has a two-partition of maximum load at most \(d\), which
implies \(T\Delta_2^+(p/T)\le d\).

If \(k=1\), then \(C\ge\alpha+p\).  Greedy balancing of any tail with the
same \((T,p)\) has maximum load at most
\[
  \frac{T+p}{2}
  =
  \frac{1-C+p}{2}
  \le
  \frac{1-\alpha}{2}
  \le d.
\]

Now assume \(k\ge2\).  For every \(h\ge2\), the exact two-agent formula
implies
\[
  T\Delta_2^+(p/T)\le G_{2,h}(T,p).
  \tag{13}
\]
To verify the exceptional ranges, put \(x=p/T\).  For \(x\le1/5\), the
integer \(\lfloor(1-x)/(2x)\rfloor\) in the last line of
Proposition~\ref{prop:q2formula} is at least two.  For
\(1/5<x\le1/3\), the three exceptional branches
\[
  \frac{3(1-x)}4,\qquad
  \frac{2+3x}{5},\qquad
  2x
\]
are each at most \(3x\).  For \(x>1/3\),
\(\Delta_2^+(x)\le1<3x\).  Since the increasing branch of
\(g_{2,h}\) is at least \(3x\), (13) follows.

For this value of \(k\), the minimum possible value of \(d\) over the
corresponding range of \(\alpha\) satisfies
\[
  d\ge
  \frac{(k+1)(k+2)}{3k^2+7k+5}
  \ge
  \frac{k+3}{3k+7}
  \tag{14}
\]
The final cross-product difference is \(k-1\ge0\).

If \(p\le d/(k+2)\), use (13) with \(h=k+1\).  The increasing branch is
at most \(d\), while \(T-p<1-d\) and (14) imply
\[
  \frac{k+3}{2(k+2)}(1-d)\le d.
\]
If \(p>d/(k+2)\), use \(h=k\).  The increasing branch is at most \(d\), and
\[
  C+p
  \ge\alpha+(k+1)p
  >
  \alpha+\frac{k+1}{k+2}d.
\]
The decreasing branch in (12) then yields
\[
  T-p<
  \frac{2(k+1)}{k+2}d,
\]
so the decreasing branch of \(G_{2,k}(T,p)\) is at most \(d\).
\end{proof}

\end{document}